\documentclass[11pt]{article}
\usepackage{chaithanya}
\title{Spectral gap of Lee-Yang Hamiltonians}
\author{Chaithanya Rayudu \thanks{\text{chaithanyarss@unm.edu}}\\
University of New Mexico 
\and
Jun Takahashi \thanks{\text{juntakahashi@issp.u-tokyo.ac.jp}}\\
University of Tokyo
}
\date{}

\begin{document}

\maketitle

\begin{abstract}
The Lee-Yang theorem and its quantum extensions state that, for a broad class of Hamiltonians on any graph, the partition function’s zeros in the complex magnetic field plane lie only on the imaginary axis.
For these Hamiltonians, we prove that under a uniform $Z$-field of any strength $h$, the ground state has a spectral gap of at least $h/4$, independent of the system size and of the coupling strengths. 
The proof uses the zero-freeness of the partition function as given by Asano and Suzuki-Fisher to show exponential decay of the imaginary-time correlations for any product of $Z$-operators.
Our result gives a polynomial-time quantum algorithm for computing the ground state energy of any Lee-Yang Hamiltonian. 
\end{abstract}

\section{Introduction}
\label{sec:introduction}

The spectral gap of a quantum Hamiltonian plays a fundamental role in determining its physical properties, especially at low temperature. 
A many-body system with a constant energy gap shows exponentially decaying spatial correlations in the ground state \cite{hastings2006spectral,nachtergaele2006lieb}, which contrasts sharply with gapless critical systems. 
Gappedness also implies stability of the phase against local perturbations \cite{bravyi2010topological,michalakis2013stability} and leads to the notion of topological phases of matter \cite{hastings2005quasiadiabatic,chen2010local} where the closing of a gap marks a quantum phase transition \cite{sachdev2011quantum,fradkin2013field}. 
While serving an essential role in understanding physical properties, rigorously proving the existence of a gap for a given Hamiltonian in general is notoriously difficult. For example, the Haldane conjecture \cite{haldane1983nonlinear} that the antiferromagnetic Heisenberg chain is gapped for integer spin remains open even after four decades, despite overwhelming numerical evidence \cite{EarlyQMC4Haldane,EarlyDMRG4Haldane,latesttasaki}. 

The spectral gap also plays an important role in quantum computation. Adiabatic quantum computation \cite{KadowakiNishimori,farhietal,LatestQA}, which is a universal model of quantum computation \cite{aharonov2007adiabatic}, works along paths of Hamiltonians where the run time scales with inverse powers of the minimum gap along a path \cite{jansen2007bounds}. An inverse polynomial gap along a path that begins at a Hamiltonian with an easily prepared ground state thus makes the ground state at its end efficiently preparable. The inverse gap also controls the cost of ground-state preparation by phase estimation \cite{lin2020near}, given an initial state with nonnegligible ground-state overlap. In one dimension, a gap further implies an area-law for the ground-state entanglement \cite{hastings2007area,arad2013area} and a polynomial-time classical algorithm for computing the ground state \cite{landau2015polynomial}. In general, however, deciding whether a family of Hamiltonians is gapped is undecidable, even for translation-invariant nearest-neighbor interactions in two dimensions \cite{cubitt2015undecidability}.

In this paper, we prove a uniform lower bound on the field-induced spectral gap for a broad class of Hamiltonians, which we refer to as Lee-Yang Hamiltonians (\cref{def:LY_Hamiltonians}), thus giving an efficient adiabatic quantum algorithm for the corresponding ground state energy problem. 
\begin{theorem*}[main result, \cref{thm:spectral_gap}]
Let $H$ be a Lee-Yang Hamiltonian on $n$ qubits and let $h > 0$. Then the ground state of $H_h = H - h \sum_{i=1}^{n} Z_i$ is nondegenerate, and the spectral gap above it is at least $h/4$.
\end{theorem*}
\noindent
The Hamiltonian family we consider acts on $n$ qubits with 2-local terms: 
\begin{definition*}[Lee-Yang Hamiltonians, \cref{def:LY_Hamiltonians}]
    \begin{align}
        H = \sum_{ij} H_{ij}
        \quad \mathrm{where}\quad H_{ij} = -w^z_{ij} Z_i Z_j + w^x_{ij} X_i X_j + w^y_{ij}Y_i Y_j + w^{xy}_{ij} X_i Y_j + w^{yx}_{ij} Y_i X_j \nonumber \quad \mathrm{and}\quad \\
    \forall ij, ~
        w^z_{ij} \;\geq\; \frac{1}{2} \Big[ \big( w^x_{ij} - w^y_{ij} \big)^2 + \big( w^{xy}_{ij} + w^{yx}_{ij} \big)^2 \Big]^{\frac{1}{2}} + \frac{1}{2} \Big[ \big( w^x_{ij} + w^y_{ij} \big)^2 + \big( w^{xy}_{ij} - w^{yx}_{ij} \big)^2 \Big]^{\frac{1}{2}}. \nonumber
    \end{align}
\end{definition*}

\subsection{Algorithmic implications}

Our result answers an open question that has attracted considerable attention in recent years about the complexity of the bipartite quantum Max-Cut problem \cite{gharibian2023guest,king2023improved,ju2025improved,apte2025improved,apte2025approximation,wong2026lee,marwaha2026complexity}. 
We make progress on this question by placing the bipartite quantum Max-Cut problem and the more general ground-state energy problem of the EPR Hamiltonian \cite{king2023improved} in {\sf BQP}, improving upon the previous {\sf StoqMA} upper bound \cite{cubitt2016complexity}. 
The quantum algorithm is simple:
Starting from the ground state of $-\sum_i Z_i$ and evolving adiabatically along the path $sH-(1-s)\sum_i Z_i$, the ground state of $H-h\sum_i Z_i$ can be prepared efficiently, to arbitrary inverse-polynomial precision, for every $h\geq 1/\operatorname{poly}(n)$. Thus, adiabatic quantum computation efficiently solves the field-perturbed problem throughout this regime. By choosing $h$ sufficiently small, we obtain an inverse-polynomial additive approximation to $E_0(H)$, placing the EPR Hamiltonian ground-state energy problem in ${\sf BQP}$.
We remark that the EPR Hamiltonian allows sign-problem free quantum Monte Carlo simulations, which leaves reasonable hope for containment in {\sf BPP} \cite{marwaha2026complexity,takahashi2024,rayudu2025fast,sandvik2010computational}. 
By contrast, the Lee-Yang class also contains Hamiltonians for which there is no known way to remove the sign problem. 
One such family, the phase-shifted EPR Hamiltonians, was introduced in \cite{wong2026lee} as a candidate for quantum advantage in ground-state energy estimation. Our spectral gap result implies that it is indeed possible to efficiently estimate the ground state energies of these Hamiltonians using a quantum computer, while currently there are no known candidate classical algorithms for this problem.

\subsection{Related works}

One of the earliest works closest to ours in technical spirit was by Penrose and Lebowitz \cite{penrose1974exponential}, who proved that zero-freeness for classical lattice systems implies the decay of correlations.
More precisely, for systems with finite-range interactions, a zero-free region of the complex fugacity plane containing the origin implies a gap in the spectrum of the transfer matrix that persists in the thermodynamic limit. Our results are a quantum analogue of this statement, with the transfer direction replaced by imaginary time and the transfer matrix by the Gibbs operator of the Hamiltonian in the field. For translation-invariant classical systems, Dobrushin and Shlosman \cite{dobrushin1987completely} later proved the equivalence of the decay of correlations with the zero-freeness of the partition function, regardless of the volume and the boundary conditions. Closer to our work, Fr{\"o}hlich and Rodriguez \cite{frohlich2012some,frohlich2017cluster} proved the analyticity of the connected correlations of classical ferromagnets throughout the zero-free region of the field given by the Lee-Yang property, and derived from it their exponential clustering at nonzero field. They remarked that similar results can be proven for the Duhamel correlations of quantum lattice systems. For quantum Heisenberg models in site-dependent fields bounded away from zero, Bj{\"o}rnberg and Ueltschi \cite{bjornberg2015decay} proved exponential decay of the correlations transverse to the field, in space and in imaginary time, at every temperature. Their proof uses random-loop representations \cite{ueltschi2013random}, though they note that some of these results can also be obtained from the Lee-Yang theorem, and the observables they cover are complementary to the diagonal products of \cref{thm:Z_decay}. 
Note that our result on the correlation decay does not only change the observables, but also proves the statement for {\it all} diagonal operators. 

More broadly for quantum many-body systems, cluster expansions and the zeros of the partition function have been used to prove various results. Above a threshold temperature, where the expansions converge for arbitrary local Hamiltonians, cluster expansions establish the decay of correlations of Gibbs states and the locality of temperature \cite{kliesch2014locality}, and efficient approximation algorithms for quantum partition functions \cite{mann2021efficient}. In recent years, they have also given the separability and efficient preparation of high-temperature Gibbs states \cite{bakshi2024high}, spectral gaps and fast mixing of Lindbladian Gibbs samplers beyond geometric locality \cite{bergamaschi2026fast}, and classical algorithms for thermal expectation values of the SYK model \cite{zlokapa2026rigorous,zlokapa2026syk}. Harrow, Mehraban, and Soleimanifar \cite{harrow2020classical} connected the complex zeros of quantum partition functions to the decay of correlations and to approximation algorithms, asking whether quantum many-body systems that satisfy the Lee-Yang theorem are computationally easy. In particular, they obtained a quasipolynomial-time classical algorithm, valid at every temperature, for the partition function of anisotropic Heisenberg ferromagnets in a field, a strict subset of the Lee-Yang Hamiltonians.

More recently, Wong, Bravyi, Gosset, and Liu \cite{wong2026lee} investigated the connection between the Lee-Yang tensors and Hamiltonian complexity. They study quantum states and operators that are Lee-Yang tensors of a given radius, tensors whose entries in the computational basis generate a polynomial, in one complex variable for each tensor index, that is zero-free on the polydisc of that radius. They review the closure properties of this class, show that states of radius above one can be prepared by circuits of quasipolynomial size, and prove an analogue of the Perron-Frobenius theorem, stating that a Hermitian operator of radius above one has a nondegenerate dominant eigenvalue. They also conjecture, with numerical support, that the spectral gap of the EPR-like Hamiltonians which would provide an efficient adiabatic algorithm for the bipartite quantum Max-Cut problem with uniform weights.

\subsection{Proof overview}

The starting point of our proof is the Lee-Yang circle theorem as stated in \cref{sec:LY_Hamiltonians}. The circle theorem we use applies to the Trotterized version $\ZZ_\epsilon$ (\cref{lem:zero_free}), which says that, as a multivariate polynomials in one complex variable for each site in each Trotter layer, $\ZZ_\epsilon$ is zero-free on the open unit polydisc, so $\log \ZZ_\epsilon$ is analytic in the same region. 
Note that the circle theorem we use is the one in its fullest generality without dropping the variables from the middle layers of the Trotterization. In our case, it is essential that we keep all the variables from all the layers as that is what enables us to prove the exponential decay. It is quite common to state the theorem where the the variables from the middle layers are contracted away since this allows one to take the continuum limit $\epsilon \rightarrow 0$ and state the theorem without any Trotterization. We do not know if this version of the theorem is sufficient to prove the spectral gap. The following proof overview is in the reverse order of how it is in the paper as it will helps motivate the starting direction we take.

A key step in our proof is to prove the decay of correlations in the imaginary-time direction (\cref{thm:Z_decay}) which then implies the spectral gap (\cref{thm:spectral_gap}). For every nonempty $S \subseteq [n]$, the connected correlator of $Z_S = \prod_{i \in S} Z_i$ between the imaginary times $0$ and $t$ decays as $\exp\{-\frac{h}{2} \min(t,\, \beta - t)\}$ for every inverse temperature $\beta$, with a prefactor polynomial in $n$, $\beta$, the total coupling strength, and $1/h$. An essential feature of this decay is that the rate is same for all $S$. Since the products $Z_S$ span the algebra of diagonal operators, the exponential decay applies to every diagonal operator. Taking $\beta \rightarrow \infty$, this implies the spectral gap.

To prove the decay of exact correlations, we first prove the decay of Trotterized correlations. In the Trotterized partition function $\ZZ_\epsilon$ as a polynomial, inserting $Z_S$ at a Trotter layer amounts to flipping the signs of the layer variables on $S$. The Trotterized correlator is then dependent on the values of $\log \ZZ_\epsilon$ at four reflected points of the polydisc (\cref{lem:log_derivative,lem:correlator_reduction}). A key ingredient in our proof to show decay of the Trotterized correlator is a support lemma for the Taylor coefficients of $\log \ZZ_\epsilon$ (\cref{lem:taylor_support}) that says every monomial with a non-zero Taylor coefficient must form a connected cluster in the imaginary-time direction, and therefore monomials containing variables from far apart layers must have a high degree. Combining the analyticity of $\log \ZZ_\epsilon$ and the support lemma (\cref{lem:taylor_support}) gives an exponential decay of the Trotterized correlator in the number of layers between the insertions (\cref{prop:fixed_eps}). The decay rate is the same for all $S$. The bound obtained this way (\cref{prop:fixed_eps}) carries a prefactor of $(\epsilon h)^{-1}$, so it does not survive the continuum limit $\epsilon \to 0$. Instead, we need two additional steps to convert the exponential decay from the Trotterized correlator to the exact one. The first is an error bound between the Trotterized and the exact correlation function (\cref{lem:trotter_compare}), and the second is an adaptive choice of the Trotter step size that shrinks exponentially in the separation in the correlator (\cref{thm:Z_decay}).

\subsection{Organization}

The remainder of the paper is organized as follows. \cref{sec:LY_Hamiltonians} defines the Lee-Yang Hamiltonians and states the zero-freeness theorem and some relevant bounds of the Trotterized partition function. \cref{sec:higher_order} proves the decay of correlations. \cref{sec:spectral_gap} proves the spectral gap from the decay of correlations. \cref{app:trotter} contains the elementary Trotter bounds and \cref{app:trotter_comparison} contains the Trotter error bounds for the correlators. 

\section{Lee-Yang Hamiltonians}
\label{sec:LY_Hamiltonians}

Here we recall the definition of Lee-Yang Hamiltonians, the class of Heisenberg-type quantum Hamiltonians for which the Lee-Yang circle theorem holds \cite{asano1970theorems,suzuki1971zeros}, the statement of the circle theorem and some bounds on the size of the partition function for these Hamiltonians. For simplicity, we only consider the Hamiltonians without any a priori 1-local fields.
\begin{definition}[Lee-Yang Hamiltonians]
    \label{def:LY_Hamiltonians}
    Consider the following class of Heisenberg-type Hamiltonians over an $n$-qubit system
    \begin{align}
        \label{eq:LY_Hamiltonians}
        H = \sum_{ij} H_{ij}
        \quad \text{where}\quad H_{ij} = -w^z_{ij} Z_i Z_j + w^x_{ij} X_i X_j + w^y_{ij}Y_i Y_j + w^{xy}_{ij} X_i Y_j + w^{yx}_{ij} Y_i X_j.
    \end{align}
    We define the Lee-Yang Hamiltonians as the subclass that satisfies, for all $i, j \in [n]$,
    \begin{align}
        \label{eq:LY_condition}
        w^z_{ij} \;\geq\; \frac{1}{2} \Big[ \big( w^x_{ij} - w^y_{ij} \big)^2 + \big( w^{xy}_{ij} + w^{yx}_{ij} \big)^2 \Big]^{\frac{1}{2}} + \frac{1}{2} \Big[ \big( w^x_{ij} + w^y_{ij} \big)^2 + \big( w^{xy}_{ij} - w^{yx}_{ij} \big)^2 \Big]^{\frac{1}{2}}.
    \end{align}
\end{definition}
\noindent Without the cross couplings, $w^{xy}_{ij} = w^{yx}_{ij} = 0$, the condition \cref{eq:LY_condition} reduces to $w^z_{ij} \geq \max\big( |w^x_{ij}|,\, |w^y_{ij}| \big)$.

\begin{lemma}[\cite{suzuki1971zeros}]
    \label{lem:zero_free}
    For Lee-Yang Hamiltonians, the Trotterized partition function 
    \begin{align}
        \ZZ_\epsilon (\bm{y}) = \ZZ_\epsilon\left(\{y_{i,m}\}_{i \in \mathbb{Z}_n, m \in \mathbb{Z}_M}\right) 
        &\coloneqq \mathrm{Tr}\left[\prod_{m \in \mathbb{Z}_M}  L \left(\prod_{i=1}^n \left(\ketbra{0}{0}_i + y_{i,m} \ketbra{1}{1}_i\right)\right)L^{\dagger}\right] \label{eq:Trotter_Z_expansion_1}\\[6pt]
        &= \sum_{\bm{x}_0, \ldots, \bm{x}_{M-1}} \; \prod_{m \in \mathbb{Z}_M} \Big[ \langle \bm{x}_m |\, L^\dagger L\, | \bm{x}_{m+1} \rangle \prod_{j} y_{j,m}^{(\bm{x}_m)_j} \Big], \label{eq:Trotter_Z_expansion_2}
    \end{align}
    where $\epsilon = \beta/M$ and 
    \begin{align}
    \label{eq:trotter_operator}
        L \coloneqq \prod_{ij} e^{-\frac{\epsilon}{2} H_{ij}},
    \end{align}
    has no zeros, i.e. $\ZZ_\epsilon (\bm{y}) \neq 0$, when $|y_{i,m}| < 1$ for all $i \in \mathbb{Z}_n, m \in \mathbb{Z}_M$. Here, $|\bm{x}_m\rangle$ is the $Z$-basis state inserted at the $m$-th Trotter slice.
\end{lemma}

Throughout the paper, we study the Lee-Yang Hamiltonians in a uniform $Z$-field of strength $h > 0$,
\begin{align}
\label{eq:LY_with_field_Hamiltonian}
H_h := H - h \sum_i Z_i.
\end{align}
Define
\begin{align}
\label{eq:Gamma_def}
    \WW := \sum_{ij} w^z_{ij} \quad \text{and} \quad \Gamma := 3\, \WW + n h,
\end{align}
where $\Gamma$ bounds the operator norm of $H_h$ since $\|H_{ij}\| \le 3\, w^z_{ij}$ and $\|H_h\| \le 3\, \WW + nh = \Gamma$.
The size of the partition function is controlled by the following bounds which are proved in \cref{app:trotter}.

\begin{lemma}
\label{lem:size_bounds}
Define $\WW := \sum_{ij} w^z_{ij}$. The Trotterized partition function of \cref{lem:zero_free} satisfies
\begin{align}
    \label{eq:size_bounds}
    \ZZ_{\epsilon}(\bm{0}) = \| L \ket{\bm{0}} \|^{2M},\quad
    |\log(\ZZ_\epsilon(\bm{0}))| \leq 3\beta \WW, \quad  \text{and} \quad \sup_{|y_{i,m}|\leq 1} |\ZZ_\epsilon(\bm{y})| \leq 2^n e^{3\beta \WW}.
\end{align}
\end{lemma}

\section{Decay of correlations from Lee-Yang zero-freeness}
\label{sec:higher_order}

In this section, we prove that the two-point correlation functions of products of the $Z$-fields decay exponentially in the imaginary-time direction for the Hamiltonians $H_h$, at a rate proportional to $h$. Let $\ket{l}$ be the eigenstates of $H_h$ with the corresponding eigenvalues $E_l$: $H_h \ket{l} = E_l \ket{l}$, where $E_0 < E_1 \leq E_2 \cdots \leq E_{2^n-1}$. The inequality $E_0 < E_1$ is foreshadowing the fact that the ground state of $H_h$ is nondegenerate, which we prove in \cref{sec:spectral_gap}. The partition functions and the thermal states are
\begin{align}
    \ZZ = \mathrm{Tr}\left[e^{-\beta H_h}\right], \quad \ZZ_\star := e^{\beta E_0}\, \ZZ = \sum_l e^{-\beta (E_l - E_0)}, \quad \text{and}\quad \rho := \frac{1}{\ZZ}e^{-\beta H_h}.
\end{align}

For any operators $A$ and $B$ and $t \in [0, \beta]$, the imaginary-time two-point function is
\begin{align}
    \label{eq:two_point_def}
    \langle A(0)\, B(t) \rangle_\beta := \mathrm{Tr}\left[ \rho\, A\, e^{-t H_h}\, B\, e^{t H_h} \right].
\end{align}
For a nonempty subset $S \subseteq [n]$ of the sites, write $Z_S := \prod_{i \in S} Z_i$ for the product of the $Z$-fields on $S$. In this paper, we are interested in the two-point connected correlation function
\begin{align}
\label{eq:fS_def}
    f_S(t) := \braket{Z_S(0)\, Z_S(t)}_\beta - \braket{Z_S}_\beta^2 ,
\end{align}
which we prove decays exponentially in $\min\{t, \beta-t\}$ at a rate at least proportional to $h$, for every $S \subseteq [n]$. First, we show that the Trotterized correlation function (\cref{eq:Trotter_two-point_correlation_function}) decays exponentially using the zero-freeness of the Trotterized partition function (\cref{lem:zero_free}). The exponential decay is then extended to the exact correlation function using a Trotter error bound (\cref{lem:trotter_compare}) and an adaptive choice of the Trotter step size (\cref{thm:Z_decay}). The reason we require these additional steps is that taking the Trotter step size $\epsilon \rightarrow 0$ does not simply give us the exponential decay that we want on the exact correlation function.

\subsection{Decay of Trotterized correlation function}

Define the Trotterized correlation function as
\begin{align}
\label{eq:Trotter_two-point_correlation_function}
    f_{S,\epsilon}(k) := \frac{ \mathrm{Tr}\left[ \mathcal{S}^{(S)}_\epsilon\, \mathcal{S}_\epsilon^{\,k-1}\, \mathcal{S}^{(S)}_\epsilon\, \mathcal{S}_\epsilon^{\,M-k-1} \right] }{ \mathrm{Tr}\left[ \mathcal{S}_\epsilon^{\,M} \right] } - \left( \frac{ \mathrm{Tr}\left[ \mathcal{S}^{(S)}_\epsilon\, \mathcal{S}_\epsilon^{\,M-1} \right] }{ \mathrm{Tr}\left[ \mathcal{S}_\epsilon^{\,M} \right] } \right)^{\!2},
\end{align}
where
\begin{align}
    \label{eq:S_epsilon_def}
    \mathcal{S}_\epsilon := L\; e^{\epsilon h \sum_j Z_j}\; L^\dagger, \quad \mathcal{S}^{(S)}_\epsilon := L\; e^{\epsilon h \sum_j Z_j} Z_S\; L^\dagger
\end{align}
and $L$ is defined in \cref{eq:trotter_operator}. The function $f_{S,\epsilon}(k)$ is an imaginary-time correlator of the Trotterized evolution with two $Z_S$ inserted at the time layers $0$ and $k$. The following lemma expresses $f_{S,\epsilon}(k)$ through the Trotterized partition function $\ZZ_\epsilon(\bm{y})$ evaluated at the point $\bm{y} = y_\star \bm{1}=:\bm{y}_\star$, with $y_\star := e^{-2\epsilon h}$, and at its reflections in the variables at layers $0$ and $k$. The evaluation point lies well inside the zero-free region $\big\{ |y_{i,m}| < 1$ for all $i \in \mathbb{Z}_n, m \in \mathbb{Z}_M \big\}$ of \cref{lem:zero_free}, each coordinate $y_\star = e^{-2\epsilon h}$ at distance $1 - e^{-2\epsilon h}$ from the boundary.

\begin{lemma}[Correlator through sign flips]
    \label{lem:log_derivative}
    For a layer $m$, let $\sigma_m$ be the reflection of the variable $\bm{y}$ that reverses the sign of the coordinates $y_{i,m}$ for every $i \in S$ and leaves all other coordinates the same. Then, for all $1 \le k \le M-1$,
    \begin{align}
    \label{eq:correlator_log}
        f_{S,\epsilon}(k) = \left[\; \frac{ \ZZ_\epsilon( \sigma_0 \sigma_k \bm{y} ) }{ \ZZ_\epsilon( \bm{y} ) } - \frac{ \ZZ_\epsilon( \sigma_0 \bm{y} ) }{ \ZZ_\epsilon( \bm{y} ) } \cdot \frac{ \ZZ_\epsilon( \sigma_k \bm{y} ) }{ \ZZ_\epsilon( \bm{y} ) } \;\right]_{\bm{y} = \bm{y}_\star} .
    \end{align}
\end{lemma}

\begin{proof}
    Expanding each trace in \cref{eq:Trotter_two-point_correlation_function} in the $Z$ basis, by inserting a resolution of the identity between consecutive layers, brings it in to a similar form of the sum in \cref{eq:Trotter_Z_expansion_2} over configurations $\bm{x}_0, \ldots, \bm{x}_{M-1}$. The coupling factor $\langle \bm{x}_m | L^\dagger L | \bm{x}_{m+1} \rangle$ of each layer is common to all three traces and unchanged, while the diagonal factor $\prod_j y_{j,m}^{(\bm{x}_m)_j}$ is modified as follows.

    For the factor $e^{\epsilon h \sum_j Z_j}$ in $\mathcal{S}_{\epsilon}$, the corresponding layer contributes $\langle \bm{x}_m | e^{\epsilon h \sum_j Z_j} | \bm{x}_m \rangle = e^{n \epsilon h} \prod_j y_\star^{\,(\bm{x}_m)_j}$, which equals $e^{n \epsilon h}$ times the diagonal factor of \cref{eq:Trotter_Z_expansion_2} at $y_{j,m} = y_\star$. Collecting the $M$ products of $e^{n \epsilon h}$ into $e^{\beta n h}$,
    \begin{align}
        \mathrm{Tr}\big[ \mathcal{S}_\epsilon^{\,M} \big] = e^{\beta n h}\, \ZZ_\epsilon(\bm{y}_\star) .
    \end{align}
    The additional $Z_S$ in $\mathcal{S}^{(S)}_\epsilon$ multiplies the diagonal factor of its layer by $\prod_{i \in S} (-1)^{(\bm{x}_m)_i}$. Since $(-1)^{x} y^{x} = (-y)^{x}$ for $x \in \{0,1\}$, this is the substitution $y_{i,m} \mapsto -y_{i,m}$ for $i \in S$ in that layer carried out by the reflection $\sigma_m$. Placing the single insertion $Z_S$ in layer $0$, and the two insertions in layers $0$ and $k$,
    \begin{gather}
        \mathrm{Tr}\big[ \mathcal{S}^{(S)}_\epsilon\, \mathcal{S}_\epsilon^{\,M-1} \big] = e^{\beta n h}\, \ZZ_\epsilon(\sigma_0\, \bm{y}_\star) , \nonumber \\[4pt]
        \mathrm{Tr}\big[ \mathcal{S}^{(S)}_\epsilon\, \mathcal{S}_\epsilon^{\,k-1}\, \mathcal{S}^{(S)}_\epsilon\, \mathcal{S}_\epsilon^{\,M-k-1} \big] = e^{\beta n h}\, \ZZ_\epsilon(\sigma_0 \sigma_k\, \bm{y}_\star) . \label{eq:S_Z_equivalence}
    \end{gather}
    The single-insertion trace does not depend on the layer of its insertion, by cyclicity of the trace, so $\ZZ_\epsilon(\sigma_0\, \bm{y}_\star) = \ZZ_\epsilon(\sigma_k\, \bm{y}_\star)$ and the squared term in \cref{eq:Trotter_two-point_correlation_function} is the product in \cref{eq:correlator_log}. Dividing the three traces by $\mathrm{Tr}\big[ \mathcal{S}_\epsilon^{\,M} \big]$ cancels the factors $e^{\beta n h}$ and gives \cref{eq:correlator_log}.
\end{proof}

The following lemma expresses $f_{S,\epsilon}(k)$ as a combination of $\log \ZZ_\epsilon$ at $\bm{y}_\star$ and its reflections, reducing the bound on the correlator to a bound on that combination.

\begin{lemma}[Correlator through log-partition function]
\label{lem:correlator_reduction}
    Let $\Phi := \log \ZZ_\epsilon$, analytic on the zero-free polydisc $\{ \bm{y} : |y_l| < 1\ \forall l \}$ of \cref{lem:zero_free}. For $1 \le k \le M-1$, set
    \begin{align}
    \label{eq:u_k_def}
        u_k := \big[\, \Phi(\sigma_0 \sigma_k \bm{y}) + \Phi(\bm{y}) - \Phi(\sigma_0 \bm{y}) - \Phi(\sigma_k \bm{y}) \,\big]_{\bm{y} = \bm{y}_\star} .
    \end{align}
    Then $f_{S,\epsilon}(k) = A\, ( e^{u_k} - 1 )$ for a factor $A$ with $|A| \le 1$. Moreover $|f_{S,\epsilon}(k)| \le 2$, and $|f_{S,\epsilon}(k)| \le 2\, |u_k|$ whenever $|u_k| \le 1$.
\end{lemma}

\begin{proof}
    Since $\ZZ_\epsilon$ is zero-free on the polydisc, $\Phi = \log \ZZ_\epsilon$ is analytic there. Each reflection $\sigma_m$ maps the polydisc to itself, so the reflected points $\sigma_0\, \bm{y}_\star$, $\sigma_k\, \bm{y}_\star$, and $\sigma_0 \sigma_k\, \bm{y}_\star$ lie in it and \cref{eq:u_k_def} is well defined.

    Write $A := \big[ \ZZ_\epsilon(\sigma_0 \bm{y})\, \ZZ_\epsilon(\sigma_k \bm{y}) / \ZZ_\epsilon(\bm{y})^2 \big]_{\bm{y} = \bm{y}_\star}$ for the second term of \cref{eq:correlator_log}. Exponentiating \cref{eq:u_k_def} gives the first term as $\big[ \ZZ_\epsilon(\sigma_0 \sigma_k \bm{y}) / \ZZ_\epsilon(\bm{y}) \big]_{\bm{y} = \bm{y}_\star} = A\, e^{u_k}$, so \cref{lem:log_derivative} yields $f_{S,\epsilon}(k) = A\, (e^{u_k} - 1)$. The two factors of $A$ are single-insertion Trotterized moments, and by \cref{lem:moment_bounds}, of modulus at most one $|A| \le 1$. For the two bounds, $f_{S,\epsilon}(k)$ is a difference of two Trotterized moments, each of modulus at most one by \cref{lem:moment_bounds}, so $|f_{S,\epsilon}(k)| \le 2$. And if $|u_k| \le 1$, then $|e^{u_k} - 1| \le 2\, |u_k|$, hence $|f_{S,\epsilon}(k)| \le |A|\, 2\, |u_k| \le 2\, |u_k|$.
\end{proof}

By \cref{lem:correlator_reduction}, bounding $f_{S,\epsilon}(k)$ reduces to bounding $u_k$. As the proof of \cref{prop:fixed_eps} shows, the only terms in the Taylor series of $\log \ZZ_\epsilon(\bm{y})$ that contribute to $u_k$ are those that contain variables at both layer $0$ and layer $k$. Given this, the following \cref{lem:taylor_support} is a key ingredient in upper bounding $f_{S,\epsilon}(k)$. It says that any monomial with a non-zero coefficient in the Taylor series of $\log \ZZ_\epsilon(\bm{y})$ must be a connected cluster of variables in the imaginary-time direction, and therefore any monomial that contains variables from far apart imaginary-time layers must have a high degree.

\begin{lemma}[Support of the Taylor coefficients]
    \label{lem:taylor_support}
    Let $\log \ZZ_\epsilon(\bm{y}) = \sum_{\bm{q}} a_{\bm{q}}\, \bm{y}^{\bm{q}}$ be the Taylor series expansion around $\bm{y} = 0$, where $\bm{q} = (q_{j,m})_{j \in \mathbb{Z}_n,\, m \in \mathbb{Z}_M}$ is a multi-index nonnegative integer assignment to variables $\bm{y} = (y_{j,m})_{j \in \mathbb{Z}_n,\, m \in \mathbb{Z}_M}$ as
    \begin{align}
        \bm{y}^{\bm{q}} := \prod_{(j,m)} y_{j,m}^{\, q_{j,m}}.
    \end{align}
    Let 
    \begin{align}
        \mathrm{supp}(\bm{q}) := \left\{ (j,m) : q_{j,m} \ge 1 \right\}, \quad |\bm{q}| := \sum_{(j,m)} q_{j,m}.
    \end{align}
    If $a_{\bm{q}} \neq 0$ and $\mathrm{supp}(\bm{q})$ contains a location in layer $0$ and a location in layer $k$, then
    \begin{align}
        |\bm{q}| \;\ge\; \min(k,\, M - k) + 1.
    \end{align}
\end{lemma}

\begin{proof}
    We use two elementary observations. First, the coefficient of $\bm{y}^{\bm{q}}$ in the Taylor series of $\log \ZZ_\epsilon(\bm{y})$ is unchanged if every variable outside $\mathrm{supp}(\bm{q})$ is set to zero, since doing so only deletes the monomials containing those variables. We may therefore work with the restricted polynomial $\tilde{\ZZ}$, obtained from $\ZZ_\epsilon$ by setting $y_l = 0$ for all $l \notin \mathrm{supp}(\bm{q})$.

    Second, call a layer $m$ \emph{dead} if no site of layer $m$ lies in $\mathrm{supp}(\bm{q})$. In $\tilde{\ZZ}$ every dead layer carries the rank one projector $\ketbra{\bm{0}}{\bm{0}} =\ketbra{0}{0}^{\otimes n} $ and a rank-one projector inserted into a trace cuts the imaginary-time direction at that layer. Concretely, if $m_1$ and $m_2$ are two dead layers, cutting the imaginary-time circle at both of them splits the cyclic trace as 
    \begin{align}
        \tilde{\ZZ}(\bm{y}_{\text{rest}}) = \big\langle \bm{0} \big|\, \Pi_{1}\, \big| \bm{0} \big\rangle \cdot \big\langle \bm{0} \big|\, \Pi_{2}\, \big| \bm{0} \big\rangle =: F_1\!\left( \bm{y}_{\mathrm{arc}_1} \right) F_2\!\left( \bm{y}_{\mathrm{arc}_2} \right),
    \end{align}
    where $\Pi_1$ and $\Pi_2$ are the partial operator products between the two cuts, and $F_1$, $F_2$ are polynomials in the disjoint variable sets present in the two arcs of the imaginary-time circle. $\log(\tilde{\ZZ})$ is analytic in the same region where $\log(\ZZ_\epsilon)$ is analytic in the remaining variables $\bm{y}_{\text{rest}}$, and so are $\log F_1$ and $\log F_2$. Hence $\log \tilde{\ZZ} = \log F_1 + \log F_2$, a sum of a function of the arc-one variables and a function of the arc-two variables, contain no monomial whose support contains variables in both the arcs.

    Now suppose $a_{\bm{q}} \neq 0$ with $\mathrm{supp}(\bm{q})$ containing variables from layers $0$ and $k$. These two layers split the remaining layers into two arcs, of $k-1$ and $M-k-1$ layers. If both arcs contained a dead layer, cutting at these dead layers would place layers $0$ and $k$ in different arcs, and by the above second observation the coefficient of $\bm{y}^{\bm{q}}$ in $\log \tilde{\ZZ}$, which equals $a_{\bm{q}}$, would vanish. Hence at least one arc between layer $0$ and $k$ has no dead layer, that is, every layer of that arc has a site in $\mathrm{supp}(\bm{q})$. That arc has at least $\min(k, M-k) - 1$ layers, so together with the locations in layers $0$ and $k$, $|\bm{q}| \ge |\mathrm{supp}(\bm{q})| \ge \min(k, M-k) + 1$.
\end{proof}

We abbreviate $d_k := \min(k,\, M - k) + 1$, so that the bound of \cref{lem:taylor_support} reads $|\bm{q}| \ge d_k$. It remains to use this high-degree bound and give an exponential bound on $f_{S,\epsilon}(k)$, using the analyticity of $\log \ZZ_\epsilon$. This is done in \cref{prop:fixed_eps}. We need the following bound, in \cref{prop:fixed_eps}, on the modulus of $\log \ZZ_\epsilon$ on a polydisc lying inside the zero-free region.

\begin{lemma}[Modulus bound on the log-partition function]
\label{lem:logZ_modulus}
    Let $C_\star := \max\{ n,\, 6 \beta \WW \}$. If $\epsilon h \le 2$, then $\log \ZZ_\epsilon$, analytic on the zero-free polydisc, satisfies
    \begin{align}
        \sup_{|y_l| \le e^{-\epsilon h/2}\ \forall l} \big| \log \ZZ_\epsilon(\bm{y}) \big| \;\le\; \frac{16\, C_\star}{\epsilon h} .
    \end{align}
\end{lemma}

\begin{proof}
    Write $e^{-\epsilon h/2} =: 1 - \mu$, so that $\mu \ge \epsilon h / 4$ because $\epsilon h \le 2$. By \cref{eq:size_bounds}, the real part of $\log \ZZ_\epsilon$ obeys $\mathrm{Re} \log \ZZ_\epsilon = \log |\ZZ_\epsilon| \le n \log 2 + 3\beta \WW$ on the \emph{whole closed unit polydisc}. For the imaginary part, the Borel-Carath\'eodory inequality uses an upper bound on the real part on a disc, together with the value at the center of the disc, and turns it into a bound on the modulus on a smaller disc. For a point $\bm{y}$ with $|y_l| \le e^{-\epsilon h/2}$ for all $l$, we apply the Borel-Carath\'eodory inequality to the function $\zeta \mapsto \log \ZZ_\epsilon(\zeta\, e^{\epsilon h}\, \bm{y})$ of one complex variable, analytic for $|\zeta| < e^{-\epsilon h/2}$ by \cref{lem:zero_free} because $|\zeta\, e^{\epsilon h}\, y_l| < 1$ there. The outer radius is $e^{-\epsilon h/2}$, on which the real-part bound above applies, the inner radius is $e^{-\epsilon h}$, at which the function takes the value $\log \ZZ_\epsilon(\bm{y})$, and the center value $\log \ZZ_\epsilon(\bm{0})$ has modulus at most $3 \beta \WW$ by \cref{eq:size_bounds}, so we obtain
    \begin{align}
        \sup_{\forall|y_l|\leq e^{-\epsilon h/2}} \left| \log \ZZ_\epsilon(\bm{y}) \right| &\le \frac{2\, e^{-\epsilon h} }{e^{-\epsilon h/2}  - e^{-\epsilon h} } \left( n \log 2 + 3 \beta \WW\right) + \frac{e^{-\epsilon h/2}  + e^{-\epsilon h} }{e^{-\epsilon h/2}  - e^{-\epsilon h} } \cdot 3 \beta \WW \nonumber\\[6pt]
        &\le\; \frac{2\, ( n \log 2 + 6 \beta \WW )}{\mu} \;\le\; \frac{4\, C_\star}{\mu} \;\le\; \frac{16\, C_\star}{\epsilon h},
    \end{align}
    where the second inequality uses $e^{-\epsilon h/2} - e^{-\epsilon h} = e^{-\epsilon h/2}\, \mu$, and the last uses $\mu \ge \epsilon h / 4$.
\end{proof}

\begin{proposition}[Decay of the Trotterized correlation function]
    \label{prop:fixed_eps}
    Suppose $\epsilon h \le 2$. Then for all $1 \le k \le M-1$,
    \begin{align}
    \label{eq:fixed_eps_bound}
        \left| f_{S,\epsilon}(k) \right| \;\le\; \frac{ 128\, C_\star }{ \epsilon h }\; e^{-\epsilon h\, d_k},
    \end{align}
    where $d_k = \min(k,\, M - k) + 1$ and $C_\star = \max\{ n,\, 6 \beta \WW \}$ is as in \cref{lem:logZ_modulus}.
\end{proposition}

\begin{proof}
    By \cref{lem:correlator_reduction} it suffices to bound $u_k$, since that lemma gives $|f_{S,\epsilon}(k)| \le 2$ and $|f_{S,\epsilon}(k)| \le 2\, |u_k|$ when $|u_k| \le 1$. Along the ray $\bm{y} = \lambda \bm{y}_\star$ set
    \begin{align}
        g(\lambda) := \big[ \Phi(\sigma_0 \sigma_k \bm{y}) + \Phi(\bm{y}) - \Phi(\sigma_0 \bm{y}) - \Phi(\sigma_k \bm{y}) \big]_{\bm{y} = \lambda \bm{y}_\star} ,
    \end{align}
    with $\Phi = \log \ZZ_\epsilon$ and its Taylor series $\Phi(\bm{y}) = \sum_{\bm{q}} a_{\bm{q}}\, \bm{y}^{\bm{q}}$ on the polydisc of \cref{lem:correlator_reduction}, so that $u_k = g(1)$.
    We bound $u_k$ in two steps. First, by \cref{lem:taylor_support}, $g$ has a zero of order $d_k$ at the origin, so the maximum modulus principle on the circle $|\lambda| = e^{\epsilon h}$ turns this zero into the decay factor $e^{-\epsilon h\, d_k}$. Second, the modulus of $g$ on that circle is bounded by \cref{lem:logZ_modulus}.

    More precisely, a reflection $\sigma_m$ multiplies the monomial $\bm{y}^{\bm{q}}$ by $(-1)^{e_m(\bm{q})}$, where $e_m(\bm{q}) := \sum_{i \in S} q_{i,m}$ is the degree of $\bm{q}$ on the locations $S \times \{m\}$. Hence the coefficient of $\bm{y}^{\bm{q}}$ in $g$ is $a_{\bm{q}} \big( 1 - (-1)^{e_0(\bm{q})} \big)\big( 1 - (-1)^{e_k(\bm{q})} \big)$, which vanishes unless $e_0(\bm{q})$ and $e_k(\bm{q})$ are both odd. When both are odd, $\mathrm{supp}(\bm{q})$ contains a location in the layer $0$ and one in the layer $k$, so $|\bm{q}| \ge d_k$ by \cref{lem:taylor_support}. Substituting $\bm{y} = \lambda \bm{y}_\star$,
    \begin{align}
        g(\lambda) = \sum_{\bm{q}\, :\; e_0(\bm{q}),\, e_k(\bm{q})\ \mathrm{odd}} 4\, a_{\bm{q}}\, y_\star^{\, |\bm{q}|}\, \lambda^{|\bm{q}|} ,
    \end{align}
    a power series in $\lambda$ whose every exponent is at least $d_k$. Thus $\tilde{g}(\lambda) := \lambda^{-d_k} g(\lambda)$ is analytic on $|\lambda| < e^{2\epsilon h}$ since every reflected point has coordinates of modulus $|\lambda y_\star| < 1$. The maximum modulus principle on $|\lambda| = r := e^{\epsilon h}$ gives
    \begin{align}
    \label{eq:u_1_bound}
        |u_k| = |g(1)| = |\tilde{g}(1)| \;\le\; r^{-d_k} \max_{|\lambda| = r} |g(\lambda)| \;=\; e^{-\epsilon h\, d_k} \max_{|\lambda| = r} |g(\lambda)| .
    \end{align}

    It remains to bound $\max_{|\lambda| = r} |g(\lambda)|$. On $|\lambda| = r$ every coordinate of the four evaluation points has modulus $|\lambda y_\star| = e^{-\epsilon h} \le e^{-\epsilon h/2}$, so each of the four terms of $g$ is at most $\sup\{ |\log \ZZ_\epsilon(\bm{y})| : |y_l| \le e^{-\epsilon h/2}\ \forall l \}$ in modulus. By \cref{lem:logZ_modulus} this supremum is at most $16\, C_\star/(\epsilon h)$, so $\max_{|\lambda| = r} |g(\lambda)| \le 64\, C_\star / (\epsilon h)$, and \cref{eq:u_1_bound} gives
    \begin{align}
    \label{eq:g_bound}
        |u_k| \;\le\; B := \frac{64\, C_\star}{\epsilon h}\, e^{-\epsilon h\, d_k} .
    \end{align}

    Returning to the correlator, if $B \le 1$ then $|u_k| \le B \le 1$, and \cref{lem:correlator_reduction} gives $|f_{S,\epsilon}(k)| \le 2\, |u_k| \le 2B$. If $B > 1$ then $|f_{S,\epsilon}(k)| \le 2 < 2B$. Either way $|f_{S,\epsilon}(k)| \le 2B$, which is \cref{eq:fixed_eps_bound}.
\end{proof}

\subsection{From Trotterized to exact correlation function}

We now compare the Trotterized correlation function $f_{S,\epsilon}(k)$ with the exact correlation function $f_S(t)$. This comparison is non-trivial because of the $(\epsilon h)^{-1}$ prefactor in the decay bound of \cref{prop:fixed_eps}. Due to this, taking $\epsilon \rightarrow 0$ renders the bound useless. Therefore, we need two additional ingredients: an additive error bound on the correlation function due to Trotterization, which \cref{lem:trotter_compare} gives, and an adaptive $\epsilon$ choice that shrinks exponentially with the separation $\min(t, \beta-t)$, carried out in \cref{thm:Z_decay}.

\begin{lemma}[Trotter comparison]
\label{lem:trotter_compare}
With $C_T := 48 \left( \beta \Gamma^2 + \WW \right)$, for all integers $M \ge 2$ and $\epsilon := \beta/M$ with $\epsilon \Gamma \le \tfrac12$ and $C_T\, \epsilon \le 1$, and all integers $1 \le k \le M-1$,
\begin{align}
    \left| f_S(k \epsilon) - f_{S,\epsilon}(k) \right| \;\le\; C_T\, \epsilon.
\end{align}
\end{lemma}

We prove \cref{lem:trotter_compare} in \cref{app:trotter_comparison}.

\begin{theorem}[Decay of two-point correlations of products]
\label{thm:Z_decay}
Let $h > 0$, $\beta > 0$, and $S \subseteq [n]$ nonempty. With $C_\star = \max\{ n,\, 6 \beta \WW \}$ of \cref{prop:fixed_eps} and $C_T = 48\, ( \beta \Gamma^2 + \WW )$ of \cref{lem:trotter_compare}, set
\begin{align}
\label{eq:G_K_def}
    G := 2 \max\{ C_T,\, \Gamma,\, h,\, 1 \} \qquad \text{and} \qquad K := \frac{256\, C_\star\, G}{h}+2.
\end{align}
Then for all $t \in [0, \beta]$,
\begin{align}
\label{eq:Z_decay}
    \left| f_S(t) \right| \;\le\; K\, e^{- \frac{h}{2} \min(t,\, \beta - t)}.
\end{align}
\end{theorem}

\begin{proof}
    Throughout we write $\tau := \min(t,\, \beta - t) \in [0, \beta/2]$. For any Trotter step $\epsilon = \beta/M$ and any layer $k$, the triangle inequality gives the chain
    \begin{align}
        \label{eq:three_terms}
        |f_S(t)| \;\le\; \big| f_S(t) - f_S(k \epsilon) \big| \;+\; \big| f_S(k \epsilon) - f_{S,\epsilon}(k) \big| \;+\; \big| f_{S,\epsilon}(k) \big|.
    \end{align}
    We choose the number of Trotter steps adaptively, so that the step size $\epsilon$ is of the order $e^{-h\tau/2}/G$, and $k$ as the integer closest to $t/\epsilon$. The first two terms of \cref{eq:three_terms} are then exponentially small in $\tau$ because they are proportional to $\epsilon$, and the third is as well because the decay $e^{-\epsilon h\, d_k}$ of \cref{prop:fixed_eps} at the rate $h$ is only partially offset by the prefactor $(\epsilon h)^{-1}$, which grows like $e^{h\tau/2}$ under the above choice of $\epsilon$, leaving the overall decay rate at $h/2$.

    If $\tau < 1/G$, then the claim is true for the following reason. By \cref{lem:moment_bounds}, $|\langle Z_S(0)\, Z_S(t) \rangle_\beta| \le 1$ and $|\langle Z_S \rangle_\beta| \le 1$, hence $|f_S(t)| \le 2$. 
    By $G \geq 2h$, $C_\star \ge n \ge 1$, and $h\tau/2 \leq \frac{h}{2G} \leq 1/4$, we have $K \geq 3$ and $e^{-h\tau/2} \geq e^{-1/4} \geq 2/3$. This implies $|f_S(t)| \le 2 \leq K e^{-h\tau/2}$. We may therefore assume $\tau \ge 1/G$. 
    
    Choose the number of Trotter steps based on $\tau$ as
    \begin{align}
        M := \left\lceil \beta G\, e^{h\tau/2} \right\rceil,
    \end{align}
    and let $k$ be the integer closest to $t/\epsilon$, where $\epsilon = \beta/M$. The remainder of the proof uses three properties of these choices: the step size $\epsilon$ lies in the window
    \begin{align}
        \label{eq:eps_window}
        \frac{1}{2G}\, e^{-h\tau/2} \;\le\; \epsilon \;\le\; \frac{1}{G}\, e^{-h\tau/2},
    \end{align}
    the hypotheses of \cref{prop:fixed_eps,lem:trotter_compare} hold,
    \begin{align}
        \label{eq:hypotheses_check}
        M \ge 2, \qquad \epsilon h,\ \epsilon \Gamma,\ C_T\, \epsilon \;\le\; \tfrac12, \qquad 1 \le k \le M - 1,
    \end{align}
    and rounding from the time $t$ to $k \epsilon$ for an integer $k$ changes it by at most $\frac{\epsilon}{2}$,
    \begin{align}
        \label{eq:rounding_bound}
        | t - k \epsilon | \;\le\; \frac{\epsilon}{2} \qquad \text{and} \qquad \min\{ \epsilon k,\, \beta - \epsilon k \} \;\ge\; \tau - \frac{\epsilon}{2}.
    \end{align}
    Let us verify the first two properties, the third being immediate from the choice of $k$ as the integer closest to $t/\epsilon$. Since $\tau \le \beta/2$ and $\tau \ge 1/G$, we have $\beta G\, e^{h\tau/2} \ge \beta G \ge 2$, and the window bound \cref{eq:eps_window} follows from $\beta G\, e^{h\tau/2} \le M \le \beta G\, e^{h\tau/2} + 1 \le 2\, \beta G\, e^{h\tau/2}$. For \cref{eq:hypotheses_check}, the bound $M \ge \beta G\, e^{h\tau/2} \ge 2$ was just obtained, and the three smallness conditions follow from $\epsilon \le 1/G$ and the definition of $G$. The layer $k$ is interior because $\epsilon \le 1/G \le \tau \le t$ and $\beta - t \ge \tau \ge \epsilon$. Dividing by $\epsilon$, the ratio $t/\epsilon$ lies in $[1, M-1]$, and so does its closest integer $k$.

    The three terms of \cref{eq:three_terms} are now bounded in turn. For the first term, differentiating $Z_S(s) = e^{-s H_h}\, Z_S\, e^{s H_h}$ gives $f_S'(s) = \braket{Z_S(0)\, [Z_S, H_h](s)}_\beta$, an imaginary-time two-point function with separation $0 \le s \le \beta$, so \cref{fact:tracial_holder} gives $|f_S'(s)| \le \| Z_S \|\, \| [Z_S, H_h] \| \le 2 \| H_h \| \le 2 \Gamma$. The mean value theorem with the rounding bound \cref{eq:rounding_bound}, followed by the window bound \cref{eq:eps_window} and $\Gamma \le G$, gives
    \begin{align}
        \big| f_S(t) - f_S(k \epsilon) \big| \;\le\; \frac{\epsilon}{2}\, \max_{0\leq s\leq\beta} |f_S'(s)| \;\le\; \epsilon\, \Gamma \;\le\; \frac{\Gamma}{G}\, e^{-h\tau/2} \;\le\; e^{-h\tau/2}. \label{eq:discritization_error_bound}
    \end{align}
    For the second term, \cref{eq:hypotheses_check} allows us to apply \cref{lem:trotter_compare}, and the window bound \cref{eq:eps_window} with $C_T \le G$ gives
    \begin{align}
        \big| f_S(k \epsilon) - f_{S,\epsilon}(k) \big| \;\le\; C_T\, \epsilon \;\le\; \frac{C_T}{G}\, e^{-h\tau/2} \;\le\; e^{-h\tau/2}. \label{eq:Trotter_error_bound}
    \end{align}
    For the third term, \cref{eq:hypotheses_check} allows us to apply \cref{prop:fixed_eps}, with $\epsilon\, d_k = \min\{ \epsilon k,\, \beta - \epsilon k \} + \epsilon \ge \tau$ where the inequality is by \cref{eq:rounding_bound}. The lower bound of \cref{eq:eps_window} gives $(\epsilon h)^{-1} \le \tfrac{2G}{h}\, e^{h\tau/2}$. Hence
    \begin{align}
        \big| f_{S,\epsilon}(k) \big| \;\le\; \frac{ 128\, C_\star }{\epsilon h}\; e^{-\epsilon h\, d_k} \;\le\; 128\, C_\star\, \frac{2G}{h}\, e^{h\tau/2} \cdot \, e^{-h \tau} \;=\; \frac{256\, C_\star\, G}{h}\; e^{-h\tau/2}, \label{eq:Trotterized_2p_correlation_function_decay_bound}
    \end{align}
    in which the growth $e^{h\tau/2}$ of the prefactor and the decay $e^{-h \tau}$ combine to the decay rate $h/2$ of the claim.
    Substituting \cref{eq:discritization_error_bound,eq:Trotter_error_bound,eq:Trotterized_2p_correlation_function_decay_bound} back into \cref{eq:three_terms}, we get
    \begin{align}
        |f_S(t)| \;\le\; \Big(\frac{256\, C_\star\, G}{h} + 2\Big)\, e^{-h\tau/2},
    \end{align}
    which is the claim.
\end{proof}

\section{Spectral gap from decay of correlations}
\label{sec:spectral_gap}

In this section we use the decay of the two-point correlations of the products $Z_S$ to prove that the ground state of $H_h$ is nondegenerate and that the spectral gap above it is at least $h/4$. The proof combines, for every nonempty $S \subseteq [n]$, the spectral lower bound of \cref{lem:lehmann} for the correlation $f_S(\beta/2)$ and the decay upper bound of \cref{thm:Z_decay}, and to let $\beta \to \infty$. The decay rate $h/2$ of \cref{thm:Z_decay} holds uniformly in the size of $S$, so the gap lower bound is independent of the system size.

\begin{theorem}[Spectral gap]
\label{thm:spectral_gap}
Let $H$ be a Lee-Yang Hamiltonian and $h > 0$. Then the ground state of the Hamiltonian $H_h$ of \cref{eq:LY_with_field_Hamiltonian} is nondegenerate, and
\begin{align}
\label{eq:spectral_gap}
    E_1 - E_0 \;\ge\; \frac{h}{4} .
\end{align}
\end{theorem}

We use the following spectral representation of the correlation function with the separation $\beta/2$.

\begin{lemma}[Spectral representation]
\label{lem:lehmann}
For a Hermitian operator $A$, with $\hat{A} := A - \langle A \rangle_\beta$,
\begin{align}
\label{eq:lehmann_decomposition}
    \langle A(0)\, A(\beta/2) \rangle_\beta - \langle A \rangle_\beta^2 \;=\; \frac{1}{\ZZ} \sum_{l, l'} e^{-\beta (E_l + E_{l'})/2}\, \big| \langle l |\, \hat{A}\, | l' \rangle \big|^2 \;\ge\; 0 .
\end{align}
\end{lemma}

\begin{proof}
    By the cyclicity of the trace, $\langle A(t) \rangle_\beta = \langle A \rangle_\beta$ for all $t$, so expanding the product gives
    \begin{align}
        \langle A(0)\, A(\beta/2) \rangle_\beta - \langle A \rangle_\beta^2 \;=\; \langle \hat{A}(0)\, \hat{A}(\beta/2) \rangle_\beta .
    \end{align}
    Expanding the trace in the eigenbasis of $H_h$,
    \begin{align}
        \langle \hat{A}(0)\, \hat{A}(\beta/2) \rangle_\beta = \frac{1}{\ZZ} \sum_{l, l'} e^{-\beta E_l}\, e^{-\frac{\beta}{2}(E_{l'} - E_l)}\, \langle l |\, \hat{A}\, | l' \rangle \langle l' |\, \hat{A}\, | l \rangle = \frac{1}{\ZZ} \sum_{l, l'} e^{-\beta(E_l + E_{l'})/2}\, \big| \langle l |\, \hat{A}\, | l' \rangle \big|^2.
    \end{align}
\end{proof}

\begin{proof}[Proof of \cref{thm:spectral_gap}]
    Fix $0 < \delta < \frac{h}{4}$ and let $V$ be the span of the eigenstates of $H_h$ with $E_l \le E_0 + \delta$. We prove $\dim V = 1$; since no eigenvalue then lies in $(E_0, E_0 + \delta]$ and $\delta < \frac{h}{4}$ is arbitrary, both claims follow.

    Fix a nonempty $S \subseteq [n]$ and abbreviate $\hat{Z}_S := Z_S - \langle Z_S \rangle_\beta$. By \cref{lem:lehmann} with $A = Z_S$,
    \begin{align}
    \label{eq:gap_decomposition}
        \big\langle Z_S(0)\, Z_S(\beta/2) \big\rangle_\beta - \big\langle Z_S \big\rangle_\beta^2 \;=\; \frac{1}{\ZZ} \sum_{l, l'} e^{-\beta (E_l + E_{l'})/2}\, \big| \langle l |\, \hat{Z}_S\, | l' \rangle \big|^2 ,
    \end{align}
    a sum of nonnegative terms. For eigenstates $l, l'$ in $V$ we have $(E_l + E_{l'})/2 \le E_0 + \delta$, so the corresponding weight satisfies
    \begin{align}
        \frac{ e^{-\beta(E_l + E_{l'})/2} }{ \ZZ } \;\ge\; \frac{ e^{-\beta\delta} }{ \ZZ_\star } \;\ge\; 2^{-n}\, e^{-\beta\delta} .
    \end{align}
    Keeping such single term of \cref{eq:gap_decomposition}, whose left-hand side is $f_S(\beta/2)$, and applying \cref{thm:Z_decay} at $t = \beta/2$, which gives
    \begin{align}
        f_S(\beta/2) \;\le\; K\, e^{-h\beta/4},
    \end{align}
    we obtain
    \begin{align}
        \big| \langle l |\, \hat{Z}_S\, | l' \rangle \big|^2 \;\le\; 2^n\, K\; e^{-\beta \left( \frac{h}{4} - \delta \right)} \qquad \text{for all } l, l' \text{ in } V .
    \end{align}
    The exponent is strictly negative because $\delta < \frac{h}{4}$, and the eigenstates and eigenvalues of $H_h$ do not depend on $\beta$, so letting $\beta \to \infty$ makes the right-hand side vanish. Note that $K$ only depends on $\beta$ quadratically. For $l \neq l'$ the centering does not contribute, and we obtain
    \begin{align}
    \label{eq:gap_offdiag}
        \langle l |\, Z_S\, | l' \rangle \;=\; 0 \qquad \text{for all } l \neq l' \text{ in } V .
    \end{align}
    For every $l=l'$ case, we obtain
    \begin{align}
        \label{eq:gap_diag}
        \langle l |\, Z_S\, | l \rangle - \langle Z_S \rangle_\beta \;\xrightarrow{\;\beta \to \infty\;}\; 0 , 
    \end{align}
    implying that all states in $V$ have the same $Z_S$ expectation value. 
    Now suppose $\dim V \ge 2$ and let $\ket{a}$ and $\ket{b}$ be two orthonormal eigenstates in $V$. The set of $2^n$ different $Z_S$ operators form a linear basis of the algebra of operators diagonal in the $Z$ basis. 
    Therefore, \cref{eq:gap_offdiag,eq:gap_diag} extend by linearity to any diagonal operator $D$ in the $Z$ basis which will satisfy:
    \begin{align}
        \langle a |\, D\, | b \rangle \;=\; 0 \qquad \text{and} \qquad \langle a |\, D\, | a \rangle \;=\; \langle b |\, D\, | b \rangle .
    \end{align}
    Taking $D = \ketbra{x}{x}$ for any $Z$ basis state $\ket{x}$,
    we get
    \begin{align}
        \braket{a|x} \braket{x|b} \;=\; 0 \qquad \text{and} \qquad \left|\braket{a|x}\right|^2 \;=\; |\braket{b|x}|^2 \qquad \text{for every } x \in \{0,1\}^n .
    \end{align}
    Summing the second identity over $x$ must give $\sum_x |\braket{a|x}|^2 = 1$, so $\braket{a|x} \neq 0$ for some $x_\star$. Then $|\braket{a|x_\star}| = |\braket{b|x_\star}| > 0$, and the product $\braket{a|x_\star}\, \braket{x_\star|b}$ cannot vanish, contradicting the first identity. Hence $\dim V = 1$, which completes the proof.
\end{proof}

\section*{Acknowledgements}
The authors used AI-based tools,
including Claude Opus 4.8, Fable 5, and ChatGPT 5.5, as aids for suggesting proof ideas, obtaining technical inequalities, and improving the presentation. All proofs, calculations, and conclusions were independently verified by
both authors. 
This work was supported 
by the U.S. Department of Energy, Office of Science,
Accelerated Research in Quantum Computing, Fundamental Algorithmic Research toward Quantum Utility
(FAR-Qu) (CR), and the JSPS KAKENHI Grant Numbers JP25K17310, JP25H01391, JP25H01388, JP25K24845, and JP25K24848 (JT).

\bibliographystyle{unsrt}
\bibliography{refs}

\newpage
\appendix
\crefalias{section}{appendix}
\crefalias{subsection}{appendix}

\section{Elementary Trotter bounds}
\label{app:trotter}

In this appendix we prove the size bounds of \cref{lem:size_bounds} and the elementary bounds on the Trotterized objects that are used in the proofs of the Trotter comparison lemmas in \cref{app:trotter_comparison}. These include the one-step error bounds (\cref{lem:one_step}) and the trace-ratio bound (\cref{lem:weyl_ratio}). The proofs of \cref{lem:size_bounds,lem:one_step} rely on the norm bound
\begin{align}
\label{eq:L_norm_bound}
    \|L^{\pm 1}\| \;\le\; \prod_{ij} e^{\frac{\epsilon}{2} \|H_{ij}\|} \;\le\; e^{\frac{3}{2}\epsilon \WW},
\end{align}
which follows from the submultiplicativity of the operator norm and bound on the norm $\| H_{ij} \| \;\le\; 3\, w^z_{ij}$.

\begin{proof}[Proof of \cref{lem:size_bounds}]
    Write $D_m := \prod_{i=1}^n \left(\ketbra{0}{0}_i + y_{i,m} \ketbra{1}{1}_i\right)$ for the diagonal operator inserted at the $m$-th Trotter layer, so that by \cref{eq:Trotter_Z_expansion_1},
    \begin{align}
        \ZZ_\epsilon(\bm{y}) = \mathrm{Tr}\Big[ \prod_{m \in \mathbb{Z}_M} L\, D_m L^\dagger \Big].
    \end{align}
    At $\bm{y} = \bm{0}$, every $D_m$ is the projector $\ketbra{\bm{0}}{\bm{0}}$. Combined with $\ketbra{\bm{0}}{\bm{0}}\, L^\dagger L\, \ketbra{\bm{0}}{\bm{0}} = \| L \ket{\bm{0}} \|^2 \, \ketbra{\bm{0}}{\bm{0}}$, we get
    \begin{align}
        \ZZ_\epsilon(\bm{0}) = v^M \quad \text{where} \quad v := \| L \ket{\bm{0}} \|^2 .
    \end{align}
    
    For the second bound of \cref{eq:size_bounds}, the norm bound \cref{eq:L_norm_bound} gives
    \begin{align}
        v \;\le\; \|L\|^2 \;\le\; e^{3\epsilon\WW},
    \end{align}
    while $1 = \|L^{-1} L \ket{\bm{0}}\| \le \|L^{-1}\|\, \|L\ket{\bm{0}}\|$ gives
    \begin{align}
        v \;\ge\; \|L^{-1}\|^{-2} \;\ge\; e^{-3\epsilon\WW}.
    \end{align}
    Hence
    \begin{align}
        |\log \ZZ_\epsilon(\bm{0})| \;=\; M\, |\log v| \;\le\; 3 M \epsilon \WW \;=\; 3\beta\WW.
    \end{align}

    For the third bound, $|y_{i,m}| \le 1$ implies $\|D_m\| \le 1$. Bounding the trace by the dimension times the operator norm of the product,
    \begin{align}
        |\ZZ_\epsilon(\bm{y})| \;\le\; 2^n \prod_{m \in \mathbb{Z}_M} \|L\|\, \|D_m\|\, \|L^\dagger\| \;\le\; 2^n\, \|L\|^{2M} \;\le\; 2^n e^{3\beta\WW}.
    \end{align}
\end{proof}

\begin{lemma}[One-step Trotter bounds]
\label{lem:one_step}
Let $T = e^{-\epsilon H_h}$ be the exact imaginary-time step, $\mathcal{S} = \mathcal{S}_\epsilon$ the Trotter step of \cref{eq:S_epsilon_def}, and $\widetilde{Z}_S = L\, Z_S\, L^{-1}$ the conjugated insertion of the product $Z_S = \prod_{i \in S} Z_i$, for a nonempty $S \subseteq [n]$. Recall $\Gamma = 3\WW + nh$ from \cref{eq:Gamma_def}. If $\epsilon \Gamma \le \tfrac{1}{2}$, then
\begin{align}
    \| \mathcal{S} - T \| \;\le\; 2\, \epsilon^2 \Gamma^2 \quad \text{and} \quad \big\| \widetilde{Z}_S - Z_S \big\| \;\le\; 5\, \epsilon\, \WW.
\end{align}
\end{lemma}

\begin{proof}
    For the first bound, note that the Trotter step $\mathcal{S} = L\, e^{\epsilon h \sum_j Z_j}\, L^\dagger$ is a product of exponentials
    \begin{align}
        \mathcal{S} \;=\; \prod_k e^{-\epsilon A_k},
    \end{align}
    where the generators $A_k$ run over the $H_{ij}/2$ appearing in $L$, then $-h \sum_j Z_j$, then the $H_{ij}/2$ appearing in $L^\dagger$ in the reverse order. They satisfy
    \begin{align}
        \sum_k A_k \;=\; H_h \qquad \text{and} \qquad \sigma := \sum_k \|A_k\| \;\le\; 3\WW + nh \;=\; \Gamma.
    \end{align}
    Expanding every factor as $e^{-\epsilon A_k} = \sum_{m \ge 0} (-\epsilon A_k)^m / m!$ and gathering terms of degree zero and one in $\epsilon$, we get $\I - \epsilon H_h$. The norm of the remaining terms is at most $e^{\epsilon\sigma} - 1 - \epsilon\sigma$. Since $e^x - 1 - x \le \tfrac{1}{2} x^2 e^x$ for $x \ge 0$,
    \begin{align}
        \big\| \mathcal{S} - \I + \epsilon H_h \big\| \;\le\; e^{\epsilon\sigma} - 1 - \epsilon\sigma \;\le\; \frac{(\epsilon\Gamma)^2}{2}\, e^{\epsilon\Gamma},
    \end{align}
    and the same expansion applied to the single exponential $T = e^{-\epsilon H_h}$, with $\|H_h\| \le \sigma \le \Gamma$, bounds $\|T - \I + \epsilon H_h\|$ by the same quantity. The triangle inequality and $\epsilon\Gamma \le \tfrac12$ then give
    \begin{align}
        \|\mathcal{S} - T\| \;\le\; (\epsilon\Gamma)^2\, e^{\epsilon\Gamma} \;\le\; e^{1/2}\, (\epsilon\Gamma)^2 \;\le\; 2\, \epsilon^2 \Gamma^2 .
    \end{align}

    For the second bound, write
    \begin{align}
        \widetilde{Z}_S - Z_S \;=\; [L, Z_S]\, L^{-1}
    \end{align}
    and telescope the commutator over the factors of $L$. Only the bonds with exactly one endpoint in $S$ contribute, so
    \begin{align}
        \big\| [L, Z_S] \big\| \;\le\; \sum_{i \in S,\; j \notin S} \big\| \big[ e^{-\frac{\epsilon}{2} H_{ij}},\, Z_S \big] \big\| \prod_{(kl) \neq (ij)} \big\| e^{-\frac{\epsilon}{2} H_{kl}} \big\| .
    \end{align}
    The integral representation
    \begin{align}
        [e^{A}, Z] \;=\; \int_0^1 e^{sA}\, [A, Z]\, e^{(1-s)A}\, \mathrm{d}s
    \end{align}
    gives $\| [e^{A}, Z] \| \le \| [A, Z] \|\, e^{\|A\|}$, and for a bond with $i \in S$ and $j \notin S$ one can verify that
    \begin{align}
        [H_{ij}, Z_S] \;=\; [H_{ij}, Z_i]\, Z_{S \setminus \{i\}} \qquad \text{and} \qquad \big\| [H_{ij}, Z_i] \big\| \;\le\; 4\, w^z_{ij}.
    \end{align}
    Combining the three bounds with $\sum_{i \in S,\, j \notin S} w^z_{ij} \le \WW$,
    \begin{align}
        \big\| [L, Z_S] \big\| \;\le\; \sum_{i \in S,\; j \notin S} \frac{\epsilon}{2} \cdot 4\, w^z_{ij}\, e^{\frac{\epsilon}{2} \|H_{ij}\|} \prod_{(kl) \neq (ij)} e^{\frac{\epsilon}{2} \|H_{kl}\|} \;\le\; 2\, \epsilon\, \WW\, e^{\frac{3}{2}\epsilon\WW}.
    \end{align}
   Together with $\|L^{-1}\| \le e^{\frac{3}{2}\epsilon\WW}$ from \cref{eq:L_norm_bound} and $3 \epsilon \WW \le \epsilon\Gamma \le \tfrac12$,
    \begin{align}
        \big\| \widetilde{Z}_S - Z_S \big\| \;\le\; \big\| [L, Z_S] \big\|\, \big\| L^{-1} \big\| \;\le\; 2\, \epsilon\, \WW\, e^{3 \epsilon \WW} \;\le\; 2\, e^{1/2}\, \epsilon\, \WW \;\le\; 5\, \epsilon\, \WW .
    \end{align}
\end{proof}

\begin{lemma}[Trace-ratio bound]
\label{lem:weyl_ratio}
Let $P$ and $Q$ be Hermitian operators with $P, Q \succeq b\, \I$ for some $b > 0$ and $\| P - Q \| \le \eta$. Set $\delta := \eta / b$. Then for every integer $N \ge 1$,
\begin{align}
    e^{-N\delta} \;\le\; (1 + \delta)^{-N} \;\le\; \frac{\mathrm{Tr}[P^N]}{\mathrm{Tr}[Q^N]} \;\le\; (1 + \delta)^{N} \;\le\; e^{N\delta} .
\end{align}
\end{lemma}

\begin{proof}
    Let $\lambda_1(P) \ge \lambda_2(P) \ge \cdots$ and $\lambda_1(Q) \ge \lambda_2(Q) \ge \cdots$ be the eigenvalues in decreasing order. Weyl's inequality and $\lambda_k(Q) \ge b$ give
    \begin{align}
        \lambda_k(P) \;\le\; \lambda_k(Q) + \eta \;\le\; (1 + \delta)\, \lambda_k(Q)
    \end{align}
    for every $k$. Raising to the $N$-th power, which is monotone on the positive eigenvalues, and summing over $k$,
    \begin{align}
        \mathrm{Tr}[P^N] \;=\; \sum_k \lambda_k(P)^N \;\le\; (1+\delta)^N \sum_k \lambda_k(Q)^N \;=\; (1+\delta)^N\, \mathrm{Tr}[Q^N].
    \end{align}
    Exchanging the roles of $P$ and $Q$, now using $\lambda_k(P) \ge b$, gives $\mathrm{Tr}[Q^N] \le (1+\delta)^N\, \mathrm{Tr}[P^N]$, which is the lower bound. The outer inequalities follow from $1 + \delta \le e^{\delta}$ for $\delta \ge 0$.
\end{proof}

\section{Trotter error bounds for correlators}
\label{app:trotter_comparison}

In this appendix we prove the Trotter error bounds for correlators used in the main text, \cref{lem:trotter_compare}. Both $f_S$ and $f_{S,\epsilon}$ are built from the one- and two-insertion moments of $Z_S$, which we compare individually in \cref{lem:moment_bounds,lem:moment_compare}.

We repeatedly use the following tracial matrix H\"older inequality.

\begin{fact}[Tracial matrix H\"older inequality]
\label{fact:tracial_holder}
Let $A_0, A_1, \ldots$ be operators and $B_0, B_1, \ldots$ Hermitian positive semidefinite operators on a finite-dimensional Hilbert space, interleaved in a trace with real powers $b_0, b_1, \ldots \ge 0$ of total degree $N = \sum_j b_j > 0$. Then
\begin{align}
\label{eq:schatten_holder}
    \big| \mathrm{Tr}\big[ A_0\, B_0^{b_0}\, A_1\, B_1^{b_1} \cdots \big] \big| \;\le\; \prod_i \| A_i \| \cdot \prod_j \|B_j\|^{b_j}_{N},
\end{align}
with the Schatten norm $\|B\|_N = \left( \mathrm{Tr}\big[ B^N \big] \right)^{1/N}$.
\end{fact}

\begin{proof}
    The H\"older inequality for Schatten norms \cite{baumgartner2011inequality} states that
    \begin{align}
        \big| \mathrm{Tr}[X_1 \cdots X_r] \big| \;\le\; \prod_i \|X_i\|_{p_i} \qquad \text{whenever} \quad \sum_i \frac{1}{p_i} = 1,
    \end{align}
    where $\|X\|_p = \left( \mathrm{Tr}\big[ |X|^p \big] \right)^{1/p}$ and $\|X\|_\infty = \|X\|$ is the operator norm. We apply it to the factors $A_i$ and $B_j^{b_j}$, omitting the factors with $b_j = 0$, with the exponent $p = \infty$ for every $A_i$ and $p_j = N / b_j$ for $B_j^{b_j}$, so that $\sum_j b_j / N = 1$. Since $B_j$ is positive semidefinite,
    \begin{align}
        \big\| B_j^{b_j} \big\|_{N/b_j} \;=\; \big( \mathrm{Tr}\big[ B_j^{N} \big] \big)^{b_j / N} \;=\; \|B_j\|_N^{b_j},
    \end{align}
    which gives the claim.
\end{proof}

Throughout, $T = e^{-\epsilon H_h}$, $\mathcal{S} := \mathcal{S}_\epsilon$, and $\widetilde{Z}_S = L\, Z_S\, L^{-1}$ as in \cref{lem:one_step}, so that $\mathcal{S}^{(S)}_\epsilon = \widetilde{Z}_S\, \mathcal{S}$. Substituting this into \cref{eq:Trotter_two-point_correlation_function} and recalling the definition \cref{eq:fS_def} of $f_S$, both correlators decompose into one- and two-insertion moments,
\begin{align}
\label{eq:correlator_moments}
    f_S(t) = \mathcal{G}_2(t) - \mathcal{G}_1^2 \qquad \text{and} \qquad f_{S,\epsilon}(k) = \mathcal{G}_{2,\epsilon}(k) - \mathcal{G}_{1,\epsilon}^2 ,
\end{align}
where
\begin{align}
\label{eq:moment_def}
    \mathcal{G}_1 := \braket{Z_S}_\beta , \qquad
    \mathcal{G}_2(t) &:= \braket{Z_S(0)\, Z_S(t)}_\beta , \nonumber \\[3pt]
    \mathcal{G}_{1,\epsilon} := \frac{ \mathrm{Tr}\big[ \widetilde{Z}_S\, \mathcal{S}^{M} \big] }{ \mathrm{Tr}[ \mathcal{S}^M ] } , \qquad
    \mathcal{G}_{2,\epsilon}(k) &:= \frac{ \mathrm{Tr}\big[ \widetilde{Z}_S\, \mathcal{S}^{k}\, \widetilde{Z}_S\, \mathcal{S}^{M-k} \big] }{ \mathrm{Tr}[ \mathcal{S}^M ] }.
\end{align}
By the definition of the time-evolved two-point function \cref{eq:two_point_def} and cyclicity of the trace, together with $\mathrm{Tr}[T^M] = \ZZ$, the exact moments at the layer times $t = k\epsilon$ have the matching form
\begin{align}
\label{eq:exact_moment_form}
    \mathcal{G}_1 = \frac{ \mathrm{Tr}\big[ Z_S\, T^{M} \big] }{ \mathrm{Tr}[ T^M ] } , \qquad
    \mathcal{G}_2(k\epsilon) = \frac{ \mathrm{Tr}\big[ Z_S\, T^{k}\, Z_S\, T^{M-k} \big] }{ \mathrm{Tr}[ T^M ] } .
\end{align}

\begin{lemma}[Moment bounds]
\label{lem:moment_bounds}
The exact moments $\mathcal{G}_1$ and $\mathcal{G}_2(t)$, at any time $t \in [0, \beta]$, and the Trotterized moments $\mathcal{G}_{1,\epsilon}$ and $\mathcal{G}_{2,\epsilon}(k)$, at any layer $1 \le k \le M-1$, have modulus at most one.
\end{lemma}

\begin{proof}
    By \cref{eq:two_point_def} and cyclicity of the trace,
    \begin{align}
        \ZZ\, \mathcal{G}_2(t) \;=\; \mathrm{Tr}\big[ Z_S\, e^{-t H_h}\, Z_S\, e^{-(\beta - t) H_h} \big],
    \end{align}
    a trace of the two unit-norm insertions $Z_S$ interleaved with the powers $e^{-t H_h} = \big( e^{-H_h} \big)^{t}$ and $e^{-(\beta - t) H_h}$ of total degree $\beta$. The tracial H\"older inequality \cref{eq:schatten_holder} at total degree $N = \beta$ bounds this trace by
    \begin{align}
        \| Z_S \|^2\, \big\| e^{-H_h} \big\|_\beta^{\, t}\, \big\| e^{-H_h} \big\|_\beta^{\, \beta - t} \;=\; \mathrm{Tr}\big[ e^{-\beta H_h} \big] \;=\; \ZZ,
    \end{align}
    so $| \mathcal{G}_2(t) | \le 1$, and the same argument with a single insertion gives $| \mathcal{G}_1 | \le 1$.

    The insertions $\widetilde{Z}_S$ of the Trotterized moments in \cref{eq:moment_def} need not have unit norm, so we first remove the conjugation. Writing $D := e^{\epsilon h \sum_j Z_j}$, so that $\mathcal{S} = L D L^\dagger$, conjugation by $D^{-1/2} L^{-1}$ leaves the traces in \cref{eq:moment_def} unchanged while sending
    \begin{align}
        \widetilde{Z}_S \;\mapsto\; Z_S \qquad \text{and} \qquad \mathcal{S} \;\mapsto\; R := D^{1/2} L^\dagger L\, D^{1/2},
    \end{align}
    the former because $Z_S$ commutes with $D$. Noe that $R$ is Hermitian and positive definite. Hence
    \begin{align}
        \mathcal{G}_{1,\epsilon} = \frac{ \mathrm{Tr}\big[ Z_S\, R^{M} \big] }{ \mathrm{Tr}[ R^M ] } , \qquad
        \mathcal{G}_{2,\epsilon}(k) = \frac{ \mathrm{Tr}\big[ Z_S\, R^{k}\, Z_S\, R^{M-k} \big] }{ \mathrm{Tr}[ R^M ] } ,
    \end{align}
    which are now in exactly the same form as the exact moments in \cref{eq:exact_moment_form}, with the Hermitian positive-definite operator $R$ in place of $T$. The tracial H\"older inequality \cref{eq:schatten_holder} with total degree $N = M$ therefore applies just as it did for the exact moments and gives $| \mathcal{G}_{1,\epsilon} |,\, | \mathcal{G}_{2,\epsilon}(k) | \le 1$.
\end{proof}

\begin{lemma}[Moment comparison]
\label{lem:moment_compare}
Let $M \ge 2$ be an integer and $\epsilon := \beta/M$ with $\epsilon \Gamma \le \tfrac12$ and $48\, ( \beta \Gamma^2 + \WW )\, \epsilon \le 1$. Then, for every layer $1 \le k \le M-1$,
\begin{align}
\label{eq:moment_compare_bound}
    \big| \mathcal{G}_2(k\epsilon) - \mathcal{G}_{2,\epsilon}(k) \big| ,\;\; \big| \mathcal{G}_1 - \mathcal{G}_{1,\epsilon} \big| \;\le\; 16\, ( \beta \Gamma^2 + \WW )\, \epsilon .
\end{align}
\end{lemma}

\begin{proof}
    The idea of the proof is that the moments \cref{eq:moment_def,eq:exact_moment_form} are ratios of traces of powers of a single Hermitian positive step, namely $T$ for the exact moments and $\mathcal{S}$ for the Trotterized ones, and we bound their difference by separately controlling two replacements: the conjugation of the inserted $Z_S$ by the Trotter layer $L$, and the replacement of the Trotter step $\mathcal{S}$ by the exact step $T$.

    The replacement of $\mathcal{S}$ by $T$ is controlled by the one-step error
    \begin{align}
        \label{ing:onestep} \eta := \| \mathcal{S} - T \| \;\le\; 2\, \epsilon^2 \Gamma^2,
    \end{align}
    while the replacement of $\widetilde{Z}_S = L\, Z_S\, L^{-1}$ by $Z_S$ is controlled by
    \begin{align}
        \label{ing:commutator} \zeta := \big\| \widetilde{Z}_S - Z_S \big\|  \;\le\; 5\, \epsilon\, \WW,
    \end{align}
    both proven in \cref{lem:one_step} under the assumption $\epsilon \Gamma \le \tfrac12$.

    Before we proceed to prove the claim, let us establish some bounds that we will use. Both $\mathcal{S}$ and $T$ are Hermitian positive products of exponentials of operators whose norms sum to at most $\epsilon \Gamma$, as in the proof of \cref{lem:one_step}. The resulting operator bound $e^{-\epsilon \Gamma}\, \I \preceq \mathcal{S}, T \preceq e^{\epsilon \Gamma}\,\I$ gives $\|\mathcal{S}^{-1}\|, \|T^{-1}\| \le e^{\epsilon \Gamma}$, hence
    \begin{align}
        \label{ing:power}
        \mathrm{Tr}[\mathcal{S}^{M-1}] = \mathrm{Tr}[\mathcal{S}^M \mathcal{S}^{-1}] \le e^{\epsilon \Gamma}\, \mathrm{Tr}[\mathcal{S}^M] \quad \text{and} \quad \mathrm{Tr}[T^{M-1}] \le e^{\epsilon \Gamma}\, \mathrm{Tr}[T^M].
    \end{align}
    We also need the following inequality, given by the trace-ratio bound of \cref{lem:weyl_ratio} applied with $P = \mathcal{S}$, $Q = T$, and $b = e^{-\epsilon \Gamma}$:
    \begin{equation}
        \label{ing:trace} e^{-N \delta} \;\le\; {\displaystyle \frac{ \mathrm{Tr}[\mathcal{S}^N] }{ \mathrm{Tr}[T^N] }} \;\le\; e^{N \delta} \;\; ( 1 \le N \le M ), \quad \text{where} \quad \delta := e^{\epsilon \Gamma}\, \eta.
    \end{equation}

    We bound $\left| \mathcal{G}_{2,\epsilon}(k) - \mathcal{G}_2(k\epsilon) \right|$ by separating the mismatch in the numerators from the mismatch in the denominators, inserting the hybrid ratio $\mathcal{G}_2(k\epsilon)\, \ZZ / \mathrm{Tr}[\mathcal{S}^M]$. The first term below then compares the numerators over the common denominator $\mathrm{Tr}[\mathcal{S}^M]$, while the second isolates the ratio $\ZZ / \mathrm{Tr}[\mathcal{S}^M]$ of partition functions, controlled by \cref{ing:trace} and the bound $| \mathcal{G}_2(k\epsilon) | \le 1$ of \cref{lem:moment_bounds}:
    \begin{align}
        \big| \mathcal{G}_{2,\epsilon}(k) - \mathcal{G}_2(k\epsilon) \big|
        &\;\le\; \left| \mathcal{G}_{2,\epsilon}(k) - \mathcal{G}_2(k\epsilon)\, \frac{ \ZZ }{ \mathrm{Tr}[\mathcal{S}^M] } \right| + \left| \mathcal{G}_2(k\epsilon)\, \frac{ \ZZ }{ \mathrm{Tr}[\mathcal{S}^M] } - \mathcal{G}_2(k\epsilon) \right| \label{eq:ratio_triangle} \\[3pt]
        &\;\le\; \left| \mathcal{G}_{2,\epsilon}(k) - \mathcal{G}_2(k\epsilon)\, \frac{ \ZZ }{ \mathrm{Tr}[\mathcal{S}^M] } \right| + e^{M \delta} - 1 && \text{by \cref{ing:trace}} \label{eq:ratio_denominator}.
    \end{align}
    Using triangle inequality, we can bound the first summand in \cref{eq:ratio_denominator} as
    \begin{align}
        \Big| &\mathcal{G}_{2,\epsilon}(k)\, \mathrm{Tr}[\mathcal{S}^M] - \mathcal{G}_2(k\epsilon)\, \ZZ \Big| \nonumber\\[3pt]
        &=\; \left| \mathrm{Tr}[  \widetilde{Z}_S\ \mathcal{S}^{k} \widetilde{Z}_S \mathcal{S}^{M-k} ] - \mathrm{Tr}[ Z_S\, T^{k}\, Z_S\, T^{M-k} ] \right| \\[3pt]
        &\le\; \left| \mathrm{Tr}[ \widetilde{Z}_S\, \mathcal{S}^{k}\, \widetilde{Z}_S\, \mathcal{S}^{M-k} ] - \mathrm{Tr}[ Z_S\, \mathcal{S}^{k}\, Z_S\, \mathcal{S}^{M-k}\, ] \right| + \left| \mathrm{Tr}[ Z_S\, \mathcal{S}^{k}\, Z_S\, \mathcal{S}^{M-k} ] - \mathrm{Tr}[ Z_S\, T^{k}\, Z_S\, T^{M-k}\, ] \right| \label{eq:ratio_intermediate}.
    \end{align}
    Here the first term exchanges the conjugated insertions $\widetilde{Z}_S$ for $Z_S$, and the second exchanges the Trotter step $\mathcal{S}$ for the exact step $T$. We use \cref{ing:commutator} and the H\"older inequality \cref{eq:schatten_holder} to bound the first summand in \cref{eq:ratio_intermediate} as
    \begin{align}
        \Big| \mathrm{Tr}[ \widetilde{Z}_S\, \mathcal{S}^{k}\, \widetilde{Z}_S\, \mathcal{S}^{M-k} ] &- \mathrm{Tr}[ Z_S\, \mathcal{S}^{k}\, Z_S\, \mathcal{S}^{M-k}\, ] \Big| \nonumber\\[3pt]
        &\le\; \left| \mathrm{Tr}[ \widetilde{Z}_S\, \mathcal{S}^{k}\, (\widetilde{Z}_S - Z_S)\, \mathcal{S}^{M-k} ]\right| +\left| \mathrm{Tr}[ (\widetilde{Z}_S - Z_S)\, \mathcal{S}^{k}\, Z_S\, \mathcal{S}^{M-k}\, ] \right|\\[3pt]
        &\le\; \zeta \big( 1 + \| \widetilde{Z}_S \| \big)\, \mathrm{Tr}[\mathcal{S}^M] \le 3 \zeta\, \mathrm{Tr}[\mathcal{S}^M],
    \end{align}
    where we used $\| \widetilde{Z}_S \| \le 1 + \zeta \le 2$. For the second summand we telescope $\mathcal{S}$ into $T$ one factor at a time across the $M$ powers. Each of the $M$ resulting terms carries a single factor $\mathcal{S} - T$ together with $j$ factors of $T$ and $M-1-j$ factors of $\mathcal{S}$ for some $0 \le j \le M-1$, so \cref{eq:schatten_holder}, \cref{ing:power}, and \cref{ing:trace} give
    \begin{align}
        &\Big| \mathrm{Tr}[ Z_S\, \mathcal{S}^{k}\, Z_S\, \mathcal{S}^{M-k} ] - \mathrm{Tr}[ Z_S\, T^{k}\, Z_S\, T^{M-k}\, ] \Big|\\
        &\le\; \Big| \mathrm{Tr}[ Z_S\, (\mathcal{S}-T)\mathcal{S}^{k-1}\, Z_S\, \mathcal{S}^{M-k} ]\Big| + \Big|\mathrm{Tr}[ Z_S\, T\,\mathcal{S}^{k-1}\, Z_S\, \mathcal{S}^{M-k} ] - \mathrm{Tr}[ Z_S\, T^{k}\, Z_S\, T^{M-k}\, ] \Big|\\[3pt]
        &\le\; \eta \,\mathrm{Tr}[\mathcal{S}^{M-1}] + \Big|\mathrm{Tr}[ Z_S\, T\,\mathcal{S}^{k-1}\, Z_S\, \mathcal{S}^{M-k} ] - \mathrm{Tr}[ Z_S\, T^{k}\, Z_S\, T^{M-k}\, ] \Big| \\[3pt]
        & \vdots \nonumber\\
        &\le\; \sum_{j=0}^{M-1} \eta \,\mathrm{Tr}[T^{M-1}]^{\tfrac{j}{M-1}} \mathrm{Tr}[\mathcal{S}^{M-1}]^{\tfrac{M-1-j}{M-1}}\\[3pt]
        &\le\; \sum_{j=0}^{M-1} \eta \, e^{j\delta}\,e^{\epsilon \Gamma}\, \mathrm{Tr}[\mathcal{S}^M] \;\le\; M\eta\, e^{\epsilon \Gamma} e^{M \delta}\, \mathrm{Tr}[\mathcal{S}^M],
    \end{align}
    where the last step used $\sum_{j=0}^{M-1} e^{j\delta} \le M e^{M\delta}$.

    Combining the two summands of \cref{eq:ratio_denominator} with the two bounds just obtained,
    \begin{align}
        \big| \mathcal{G}_{2,\epsilon}(k) - \mathcal{G}_2(k\epsilon) \big| \;\le\; 3 \zeta + M \eta\, e^{\epsilon \Gamma} e^{M \delta} + \big( e^{M \delta} - 1 \big).
    \end{align}
    The single-insertion moments are compared in the same way. Inserting the factor $\ZZ / \mathrm{Tr}[\mathcal{S}^M]$ and using $\ZZ = \mathrm{Tr}[T^M]$,
    \begin{align}
        \big| \mathcal{G}_{1,\epsilon} - \mathcal{G}_1 \big|
        &\;\le\; \frac{ \Big| \mathrm{Tr}[ (\widetilde{Z}_S - Z_S)\, \mathcal{S}^M ] \Big| + \Big| \mathrm{Tr}[ Z_S\, (\mathcal{S}^M - T^M) ] \Big| }{ \mathrm{Tr}[\mathcal{S}^M] } + \left| \mathcal{G}_1 \right|\, \left| \frac{\ZZ}{\mathrm{Tr}[\mathcal{S}^M]} - 1 \right| \\[4pt]
        &\;\le\; \zeta + M \eta\, e^{\epsilon \Gamma} e^{M \delta} + \big( e^{M \delta} - 1 \big),
    \end{align}
    the first numerator summand by \cref{ing:commutator} and \cref{eq:schatten_holder}, the second by telescoping $\mathcal{S}$ into $T$ across its $M$ factors as above, and the last term by \cref{lem:moment_bounds} and \cref{ing:trace}.

    It remains to insert the numerical bounds. Recall that $M = \beta / \epsilon$. The assumption $\epsilon \Gamma \le \tfrac{1}{2}$ gives $e^{\epsilon \Gamma} \le e^{1/2} \le 2$. By \cref{ing:onestep} and $( \beta \Gamma^2 + \WW )\, \epsilon \le \tfrac{1}{48}$, we have
    \begin{align}
        M \eta \;\le\; 2\, \beta \Gamma^2 \epsilon, \qquad M \delta = e^{\epsilon \Gamma} M \eta \;\le\; 4\, \beta \Gamma^2 \epsilon \;\le\; \tfrac{1}{12} < 1.
    \end{align}
    In particular $e^{M \delta} \le 2$. It follows that $M \eta\, e^{\epsilon \Gamma} e^{M \delta} \le 8\, \beta \Gamma^2 \epsilon$, and that $e^{M \delta} - 1 \le 2 M \delta \le 8\, \beta \Gamma^2 \epsilon$. Together with $\zeta \le 5\, \WW \epsilon$ from \cref{ing:commutator} (so that $3 \zeta \le 15\, \WW \epsilon$), both differences obey the common bound
    \begin{align}
        \big| \mathcal{G}_2(k\epsilon) - \mathcal{G}_{2,\epsilon}(k) \big|,\ \big| \mathcal{G}_1 - \mathcal{G}_{1,\epsilon} \big| \;\le\; 15\, \WW \epsilon + 16\, \beta \Gamma^2 \epsilon \;\le\; 16\, ( \beta \Gamma^2 + \WW )\, \epsilon,
    \end{align}
    which is \cref{eq:moment_compare_bound}.
\end{proof}

\begin{proof}[Proof of \cref{lem:trotter_compare}]
    Factoring the difference of squares in \cref{eq:correlator_moments} and using $| \mathcal{G}_{1,\epsilon} |, | \mathcal{G}_1 | \le 1$ from \cref{lem:moment_bounds}, we get
    \begin{align}
        \left| f_{S,\epsilon}(k) - f_S(k \epsilon) \right| &\;\le\; \big| \mathcal{G}_{2,\epsilon}(k) - \mathcal{G}_2(k\epsilon) \big| + \big| \mathcal{G}_{1,\epsilon} + \mathcal{G}_1 \big|\, \big| \mathcal{G}_{1,\epsilon} - \mathcal{G}_1 \big| \\[5pt]
        &\;\le\; \big| \mathcal{G}_{2,\epsilon}(k) - \mathcal{G}_2(k\epsilon) \big| + 2 \big| \mathcal{G}_{1,\epsilon} - \mathcal{G}_1 \big| ,
    \end{align}
    and \cref{lem:moment_compare} bounds each difference by $b := 16\, ( \beta \Gamma^2 + \WW )\, \epsilon$. Hence, by the difference-of-squares bound,
    \begin{align}
        \big| f_{S,\epsilon}(k) - f_S(k \epsilon) \big| \;\le\; b + 2 b = 3 b \;=\; 48\, ( \beta \Gamma^2 + \WW )\, \epsilon = C_T\, \epsilon,
    \end{align}
    which is the claim.
\end{proof}

\end{document}